\documentclass[final]{l4dc2026}

\title[Adaptive LQR for Linear Systems with Unknown Time-Varying Parameters]{Stability of Certainty-Equivalent Adaptive LQR for Linear Systems with Unknown Time-Varying Parameters}

\usepackage{times}
\usepackage{centernot}
\usepackage[shortlabels]{enumitem}
\usepackage{doi}

\newtheorem{assumption}{Assumption}
\newcommand{\ubar}[1]{\text{\b{$#1$}}}

\coltauthor{\Name{Marcell Bartos}$^{1,2}$ \Email{mbartos@ethz.ch}\\
 \Name{Johannes Köhler}$^{1,3}$ \Email{j.kohler@imperial.ac.uk}\\
 \Name{Florian Dörfler}$^2$\Email{doerfler@control.ee.ethz.ch}\\
 \Name{Melanie N. Zeilinger}$^1$ \Email{mzeilinger@ethz.ch}\\
 \addr $^1$: Institute for Dynamic Systems and Control, ETH Zürich, Zürich, Switzerland \\
 $^2$: Automatic Control Laboratory, ETH Zürich, Zürich, Switzerland \\
 $^3$: Department of Mechanical Engineering, Imperial College London, London, UK
}

\begin{document}

\maketitle

\begin{abstract}%

Standard model-based control design deteriorates when the system dynamics change during operation.
To overcome this challenge, online and adaptive methods have been proposed in the literature. In this work, we consider the class of discrete-time linear systems with unknown time-varying parameters. We propose a simple, modular, and computationally tractable approach by combining two classical and well-known building blocks from estimation and control: the least mean square filter and the certainty-equivalent linear quadratic regulator. Despite both building blocks being simple and off-the-shelf, our analysis shows that they can be seamlessly combined to a powerful pipeline with stability guarantees. Namely, finite-gain $\ell^2$-stability of the closed-loop interconnection of the unknown system, the parameter estimator, and the controller is proven, despite the presence of unknown disturbances and time-varying parametric uncertainties. Real-world applicability of the proposed algorithm is showcased by simulations carried out on a nonlinear planar quadrotor.%
\end{abstract}

\begin{keywords}%
adaptive control, parametric uncertainty, online learning for control, unknown systems, linear time-varying systems, linear quadratic regulator%
\end{keywords}

\section{Introduction}

Many successful control designs, such as the linear quadratic regulator (LQR)~\citep{anderson2007optimal}, or model predictive control~\citep{rawlings2020model}, require a known model of the controlled system. However, in real world scenarios, it is often cumbersome and expensive to derive accurate models from first principles.
\par Data-driven methods~\citep{hou2013model} have been proposed to address this issue, leveraging an offline-collected dataset to either directly synthesize the controller~\citep{de2019formulas,van2025data}, or indirectly, via a system identification step~\citep{ljung1998system,van2012subspace}. To account for the remaining model uncertainty, usually a robust control approach is considered~\citep{berberich2020data, van2020noisy}, which is only feasible if the error is small. In addition, offline designs fail if the system behavior changes online, or if no stable open-loop experiments can be conducted. To overcome these challenges, online and adaptive algorithms are needed that update the model and the policy to improve performance over time. 
\par Reinforcement learning~\citep{sutton1998reinforcement} is another popular approach to address the lack of an accurate model. These methods aim to learn an optimal policy via repeated interaction with the environment, usually in an episodic and simulated manner, leveraging the availability of a vast collection of data and increasing computational power. 
However, the vast majority of reinforcement learning methods work offline~\citep{levine2020offline}, rendering them incapable of adapting to online changes in the system dynamics.
Furthermore, they generally only achieve the desired guarantees  \textit{after training}, whereas failure (e.g. in terms of instability or arbitrarily large cost) is admissible \textit{during training}. Transferring these techniques to the online setting where learning happens through real-world experiments is challenging, as in this case transient behavior is important, and resetting the experiments at the start of every episode is not always possible.
\par In this work, we address the optimal control of discrete-time linear systems, when  the model parameters are unknown and time-varying. 

\subsection{Related Work}
\hspace{16pt}\textbf{Classical adaptive control:} The well-established field of adaptive control~\citep{astrom1994adaptive, krstic1995nonlinear, narendra2012stable} aims to stabilize systems with parametric uncertainties by continuously adapting the controller. Limitations include the fact that the parameter adaptation rule and the control policy have to be simultaneously designed and tailored to specific, often restrictive problem classes (e.g., assuming that the system is minimum phase). Additionally, optimality of the controller is rarely an objective of classical adaptive control, and most results are limited to continuous-time systems. 
For example, discrete-time gradient descent requires a good choice of the step size, unlike gradient flows in continuous time.
\par \textbf{Online learning the LQR:} In the past decade there has been an increased interest in applying online learning techniques~\citep{shalev2012online} to learn the LQR for linear time-invariant (LTI) systems. These include methods from the optimism in the face of uncertainty framework~\citep{abbasi2011regret, cohen2019learning, bartos2026optimistic};
policy gradient methods~\citep{fazel2018global, zhao2025policy}; certainty-equivalent control~\citep{wang2021exact}; and policy iteration~\citep{song2024role}. The methods of~\citet{abbasi2011regret, wang2021exact, bartos2026optimistic} focus on deriving regret bounds, while~\citet{cohen2019learning,song2024role,zhao2025policy} also certify stability of the closed loop. However, with the exception of~\citet{zhao2025policy, wang2021exact}, they carry out updates of the policy in an episodic manner, i.e., only after some (usually increasing) number of time steps. The methods in~\citet{zhao2025policy, wang2021exact} require access to a known stabilizing initial controller and work by injecting exploratory noise into the control input to ensure that the data is persistently exciting~\citep{narendra2012stable,willems2005note}. Finally, a common limitation of all these methods is that they consider the time-invariant case, i.e., they cannot handle online-changing parameters.
\par \textbf{Online learning in the time-varying case:} A few recent works have addressed the control problem of time-varying systems from the online learning perspective. The methods proposed in~\citet{qu2021stable,lin2021perturbation} come with performance guarantees, but require that the true parameter value is revealed to them every time step before selecting the control input. The methods of~\citet{minasyan2021online, gradu2023adaptive} assume that the system is open-loop stable, and~\citet{noori2025data} only consider the finite-horizon problem. Under the assumption of infrequently changing or slowly drifting dynamics, the method of~\citet{yu2023online} guarantees stability of the closed-loop system, however, it requires solving a second-order cone program every time step, where the number of constraints grows linearly with time, limiting its real-world applicability.

\subsection{Contributions}
We propose a simple, modular, and computationally tractable approach by combining two classical and well-known building blocks from estimation and control: the least mean square filter (LMS)~\citep{astrom1994adaptive} and the certainty-equivalent LQR. Although the LMS and the LQR are widely used, stability guarantees for corresponding adaptive controllers for time-varying parameters are not available in the literature.
\par In contrast with the aforementioned online-learning-based methods, the proposed method is applicable to linear systems with continuously changing parameters and operates in a non-episodic manner, i.e., by adapting the policy every time step. Moreover, unlike classical adaptive control, our approach is modular, general, uses off-the-shelf building blocks, and addresses the discrete-time setting.
\par Despite the simplicity of the proposed approach, robust stability of the closed-loop interconnection is proven, even without assuming that the data is persistently exciting.
In the case of persistently exciting data, simulation results suggest that the proposed algorithm is also able to learn the model and achieve optimality in the asymptotic sense.

\subsection{Outline and Notation}
\hspace{16pt} \textbf{Outline:}
Section~\ref{sec:problem_setup} defines the considered problem setting, and Section~\ref{sec:proposed_method} presents the proposed method as well as its stability analysis. 
Section~\ref{sec:numerical_example} contains simulation results 
on a planar nonlinear quadrotor model. 
Section~\ref{sec:conclusion} concludes the paper and provides an outlook on future directions.
\par \textbf{Notation:} Let $\mathbb{R}$ and $\mathbb{N}$ denote the set of real and nonnegative integer numbers, respectively. $I_n$ represents the $n\times n$ identity matrix, positive (semi-)definiteness of a matrix $Q$ is indicated by $Q\succ(\succeq)\ 0$, and the Kronecker product of two matrices is denoted by $\otimes$.
The 2-norm of vector $x$ is denoted by $\|x\|$, while $\|x\|_P := \sqrt{x^\top P x}$ for $P = P^\top \succ 0$. For matrix $B \in \mathbb{R}^{n\times m}$, $\|B\| $ denotes its spectral norm, and its vectorization is denoted by $\mathrm{vec}(B) \in \mathbb{R}^{nm}$, which is obtained by stacking the columns of $B$ on top of one another. For vectors $a \in \mathbb{R}^n, b\in\mathbb{R}^m$, we define $(a,b) := [a^\top\ b^\top]^\top \in \mathbb{R}^{n+m}$. For a sequence $\{a_0,a_1,...\} = \{ a_k\}_{k=0}^\infty = \{a\}$, we use $\{a\} \in \ell^1$ or $\{a\} \in \ell^2$ to denote that it is summable or square-summable, respectively.

\section{Problem Setup} \label{sec:problem_setup}

Consider the discrete-time linear time-varying (LTV) system
\begin{equation} \label{eq:System}
    x_{k+1} = A(\theta_k) x_k + B(\theta_k) u_k + w_k,
\end{equation}
with state $x_k \in \mathbb{R}^n$, input $u_k \in \mathbb{R}^m$, and unknown disturbance $w_k \in \mathbb{R}^n$ at time step $k \in \mathbb{N}$, and initial condition $x_0 \in \mathbb{R}^n$. We assume that the state can be measured directly 
and that the matrices $A: \mathbb{R}^p \rightarrow \mathbb{R}^{n \times n}$ and $B: \mathbb{R}^p \rightarrow \mathbb{R}^{n \times m}$ are known
functions of the unknown time-varying parameter $\theta_k \in \mathbb{R}^p$. Let $\Delta\theta_k:=\theta_{k+1} - \theta_k$.

\begin{assumption}[Bounded disturbances] \label{ass:W}
    The disturbances are uniformly bounded, i.e., $\exists \bar{W} \geq 0$ such that $\|w_k\| \leq \bar{W},\ \forall k \in \mathbb{N}$.
\end{assumption}
\begin{assumption}[Parameters] \label{ass:Theta}
    The parameter vector is contained in a known compact convex set: $\theta_k \in \Theta,\ \forall k \in \mathbb{N}$. The set $\Theta$ is pointwise stabilizable, i.e., $\forall \theta \in \Theta$, the pair $(A(\theta),B(\theta))$ is stabilizable. The matrices $A(\theta), B(\theta)$ depend affinely on $\theta$.
\end{assumption}

Note that the pointwise stabilizability condition in Assumption~\ref{ass:Theta} is significantly weaker than assuming the existence of a single controller that simultaneously stabilizes every element of $\Theta$.
\par Our primary goal is to control system~\eqref{eq:System} while ensuring $\ell^2$-stability of the closed-loop system with respect to the disturbances $w_k$, despite the presence of unknown time-varying parameters $\theta_k$. Definition~\ref{def:l2} provides the exact definition of $\ell^2$-stability.

\begin{definition} \textbf{(Finite-gain $\ell^2$-stability, adapted from~\citet[Definition~5.1]{khalil2002nonlinear})} \label{def:l2}
    Closed-loop system~\eqref{eq:System} under a controller $u_k = \pi_k(x_k)$ is finite-gain $\ell^2$-stable with respect to the disturbance signal $\{w\}$, if $\exists c_1,c_2 \geq 0$ such that
    \begin{equation*}
        \sum_{k=0}^T \| x_k\|^2 \leq c_1 \sum_{k=0}^T \| w_k\|^2 + c_2, \quad \forall T \in \mathbb{N},\ \forall \{w\} \in \ell^2.
    \end{equation*}
\end{definition}
If the closed-loop system is finite-gain $\ell^2$-stable, then $\lim_{k\rightarrow\infty} x_k = 0$, as long as $\{w\} \in \ell^2$.
In our analysis, the constant $c_2$ in Definition~\ref{def:l2} will depend on the parameter variation signal $\{ \Delta\theta\}$.
Additionally, the policy is to be adapted with the aim to minimize the asymptotic average cost $ \lim_{T \rightarrow \infty} \frac{1}{T}\sum_{k=0}^T (x_k^\top Q x_k + u_k^\top R u_k)$, 
where $ Q,R \succ 0$.

\section{Proposed Method} \label{sec:proposed_method}

The proposed method consists of two parts: the model learner and the policy update. The learner continuously updates an estimate $\hat{\theta}_k$ of the unknown parameter $\theta_k$ based on the most recent state transition measurement. This parameter estimate is then fed into the policy update rule, which in turn continuously updates the control policy every time step
(see Figure~\ref{fig:interconnection}).

\begin{figure}[h] 
  \centering
  \includegraphics[scale=0.95]{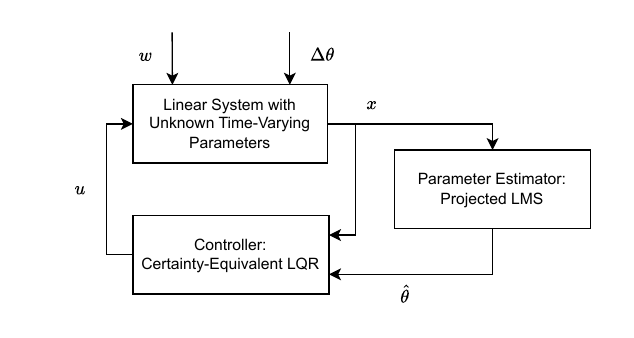}
  \caption{The interconnection of the unknown system, the parameter estimator, and the controller.} \label{fig:interconnection}
\end{figure}

\subsection{Model Learner} \label{sec:model_learner}

Before presenting the proposed model learner, we equivalently reformulate~\eqref{eq:System} in a form akin to affine linear regression:
\begin{equation} \label{eq:System_D}
    x_{k+1} = \delta(x_k,u_k) + D(x_k,u_k)\theta_k + w_k,
\end{equation}
where $\delta(x_k,u_k): \mathbb{R}^{n} \times \mathbb{R}^m \rightarrow \mathbb{R}^n$ and $D(x_k,u_k): \mathbb{R}^{n} \times \mathbb{R}^m \rightarrow \mathbb{R}^{n \times p}$ are linear and known functions. In the following, system representations~\eqref{eq:System} and~\eqref{eq:System_D} are used interchangeably, and the shorthand $D_k := D(x_k,u_k)$ will be used when convenient.
\par Define the 1-step nominal prediction of the state at time $k$ as $\hat{x}_{1|k} = A(\hat{\theta}_k)x_k + B(\hat{\theta}_k)u_k$.
The proposed model learner is the projected least mean square (LMS) estimator:
\begin{equation} \label{eq:lms}
    \hat{\theta}_k = \Pi_\Theta[\hat{\theta}_{k-1}+\mu D(x_{k-1},u_{k-1})^\top (x_k-\hat{x}_{1|k-1})],
\end{equation}
where $\Pi_\Theta[\theta] = \mathrm{argmin}_{\tilde{\theta}\in \Theta} \|\tilde{\theta}-\theta \|$ represents the
projection onto the convex set $\Theta$, and $\mu>0$ is the step size. For the analysis in the following sections, we introduce the following definitions. Let $\varphi_k$ denote the parameter estimation error at time $k$, defined as $\varphi_k := \hat{\theta}_k-\theta_k$. Given $\varphi_k$, we can define the 1-step nominal prediction error as
\begin{equation*}
    e_{1|k} := x_{k+1}-\hat{x}_{1|k}-w_k = D_k\theta_k - D_k\hat{\theta}_k = -D_k\varphi_k.
\end{equation*}
An interesting insight is that in the noiseless case ($w_k \equiv 0$), the LMS update is equivalent to a gradient descent step on the squared 1-step prediction error $\|e_{1|k}\|^2$. 
\par The (unprojected) LMS filter is a fundamental tool in recursive parameter estimation~\citep{astrom1994adaptive}. The projected version has recently been applied for parameter estimation in adaptive model predictive control~\citep{lorenzen2019robust,degner2026adaptive}.
While for general convex sets $\Theta$ the projected LMS involves solving a convex optimization problem at every time step, in the common case where $\Theta$ is a hyperrectangle, projection reduces to clipping the input vector in an elementwise manner.

An advantage of the LMS compared to the more popular recursive least squares (RLS) method is that RLS is unsuitable for the estimation of continuously changing parameters due to its inability to forget.
A popular approach to overcome this limitation of RLS is to use exponential forgetting. However, ensuring desirable behavior of the resulting adaptive controller becomes challenging, as one needs to consider how fast the parameters change and tune the forgetting factor accordingly.


\subsection{Policy Update} \label{sec:policy_update}

Given the current estimate $\hat{\theta}_k$, the policy is updated to the corresponding certainty-equivalent LQR. In our case, certainty equivalence means that the algorithm pretends that the true parameter does not change over time, and it also disregards the presence of disturbances. The policy update is given by
\begin{equation} \label{eq:lqr}
    u_k = K_\mathrm{LQR}(\hat{\theta}_k)x_k,
\end{equation}
where $K_\mathrm{LQR}(\theta)$ represents the LQR gain corresponding to the LTI system $A(\theta), B(\theta)$ and weight matrices $Q, R \succ 0$, defined as $K_\mathrm{LQR}(\theta) =-(R+B(\theta)^\top P(\theta) B(\theta))^{-1}B(\theta)^\top P(\theta) A(\theta)$. Here $P(\theta)$ denotes the unique positive definite solution to the discrete-time algebraic Riccati equation
\begin{equation} \label{eq:dare}
    P(\theta) = A(\theta)^\top P(\theta)A(\theta) - A(\theta)^\top P(\theta) B (\theta) (R + B(\theta)^\top P(\theta) B(\theta))^{-1}B(\theta)^\top P(\theta) A(\theta) + Q,
\end{equation}
and it serves as the Lyapunov function associated with the LQR:
\begin{equation} \label{eq:Lyapunov_def}
    V(x,\theta) = x^\top P(\theta) x,
\end{equation}
\par In summary, at each time step, the algorithm pretends to know the model parameter perfectly and that this parameter does not change over time, and selects the corresponding optimal infinite horizon controller.
As shown in the following sections, this results in a stable closed-loop behavior. Since both the LMS and LQR are computationally cheap, the policy can be updated at every sampling time, even at high rates. The proposed algorithm is summarized in Algorithm~\ref{alg:LMS+LQR}.

\begin{algorithm2e}
\caption{Certainty-Equivalent Adaptive LQR}
\label{alg:LMS+LQR}
\LinesNumbered
\KwIn{Parameter set $\Theta$, initial state $x_0$, initial parameter estimate $\hat{\theta}_0 \in \Theta$, step size $\mu>0$;}
\For{$k \in \mathbb{N}$}{
Update policy to the certainty-equivalent LQR: $K_k = K_\mathrm{LQR}(\hat{\theta}_k)$\;
Apply control input $u_k = K_kx_k$\;
Measure new state $x_{k+1}$\;
Update the parameter estimate $\hat{\theta}_{k+1}$ according to the projected LMS estimator~\eqref{eq:lms}\;
}
\end{algorithm2e}

\subsection{Closed-loop Analysis} \label{sec:analysis}

Although system~\eqref{eq:System} is linear, its closed-loop interconnection with the projected LMS estimator~\eqref{eq:lms} and the certainty-equivalent LQR~\eqref{eq:lqr} is nonlinear and time-varying (due to the fact that the control policy is a nonlinear function of the past and the system is time-varying). Consequently, we resort to a Lyapunov-type stability analysis. 
First, we 
define the diameter of the set $\Theta$ as $d = \max_{\theta_1,\theta_2 \in \Theta}\| \theta_1-\theta_2\|$, and
restate relevant results on the LMS from the literature.

\begin{proposition}\textbf{(LMS bounds~\citep[Proposition~1]{degner2026adaptive})} \label{prop:lms}
    Let Assumptions~\ref{ass:W} and~\ref{ass:Theta} hold. Suppose that the state and the input are uniformly bounded: $ \|x_k\| \leq X,\|u_k\| \leq U,\ \forall k \in \mathbb{N}$. If the step size satisfies
    \begin{equation} \label{eq:mu_cond}
        \frac{1}{\mu} \geq \sup_{\|x\| \leq X,\|u\|\leq U}\| D(x,u) \|^2,
    \end{equation}
    then
    \begin{enumerate}[a)]
        \item the parameter estimation error $\varphi_k = \hat{\theta}_k-\theta_k$ satisfies
        \begin{equation} \label{eq:lyap_phi}
            \| \varphi_{k+1}\|^2 - \| \varphi_k\|^2 \leq -\mu \| e_{1|k}\|^2 + \mu \| w_k\|^2 + 2d\| \Delta\theta_k\|,\quad \forall k \in \mathbb{N},
        \end{equation}
        \item the distance between consecutive estimates is bounded by
        \begin{equation} \label{eq:hat_bound}
            \| \hat{\theta}_{k+1} - \hat{\theta}_k\| \leq \sqrt{\mu} \| e_{1|k} + w_k\|,\quad \forall k \in \mathbb{N},
        \end{equation}
        \item and $\forall T \in \mathbb{N}$ it holds that
        \begin{equation} \label{eq:e_bound}
            \sum_{k=0}^T \| e_{1|k}\|^2 \leq \frac{1}{\mu} \| \varphi_0\|^2 + \sum_{k=0}^T \left( \| w_k\|^2 + \frac{2d}{\mu}\| \Delta\theta_k\|\right).
        \end{equation}
    \end{enumerate}
\end{proposition}
Note that $\| \varphi_k\|^2$ can be thought of as a Lyapunov function for the estimator dynamics~\eqref{eq:lms}, and~\eqref{eq:lyap_phi} as the corresponding Lyapunov inequality. The right-hand side consists of increase terms as functions of the disturbances $w_k$ and $\Delta\theta_k$, and a semidefinite decrease term: $\|e_{1|k}\|^2 = \| D(x_k,u_k)\varphi_k\|^2 = \varphi_k^\top D_k^\top D_k \varphi_k$ with $D_k^\top D_k \succeq 0$.
\par Inequality~\eqref{eq:e_bound} provides a finite-horizon upper bound on the prediction error. Taking the limit $T \rightarrow \infty$, this result implies that the prediction error sequence is square-summable ($\{e\} \in \ell^2$), as long as $\{w\} \in \ell^2$ and $\{\Delta\theta\} \in \ell^1$. Inequality~\eqref{eq:hat_bound} provides a bound on the variation of the parameter estimates $\hat{\theta}_k$, which will be important in the following closed-loop analyses. Finally, note that Proposition~\ref{prop:lms} does not require persistency of excitation.
\par Having derived bounds and a Lyapunov inequality for the LMS estimator, we proceed to derive similar results for the policy update. We will make use of the following technical lemma:
\begin{lemma}\textbf{(Lipschitz continuity of the LQR over $\Theta$)} \label{lem:lqr_lipschitz}
    Let Assumption~\ref{ass:Theta} hold. Then $\exists L_K>0$ such that
    \begin{equation*}
        \| K_\mathrm{LQR}(\theta') - K_\mathrm{LQR}(\theta) \| \leq L_K\| \theta'-\theta\|,\quad \forall \theta, \theta' \in \Theta.
    \end{equation*}
\end{lemma}
\begin{proof}
    Lipschitz continuity can be proven by showing that the Jacobians $\partial \mathrm{vec}(P(\theta))/\partial\mathrm{vec}(A(\theta))$ and $\partial \mathrm{vec}(P(\theta))/\partial\mathrm{vec}(B(\theta))$ derived in~\citet[Proposition~2]{east2020infinite} always exist and are uniformly upper bounded. The detailed proof can be found in Appendix~\ref{app:lqr_lipschitz}.
    \vspace{-0.2cm}
\end{proof}
With the help of Lemma~\ref{lem:lqr_lipschitz}, the following result can be proven:

\begin{proposition}\textbf{(LQR Lyapunov decrease)} \label{prop:lqr}
    Let Assumptions~\ref{ass:W} and~\ref{ass:Theta} hold. Suppose that the state is uniformly bounded: $ \|x_k\| \leq X,\ \forall k \in \mathbb{N}$. Consider any sequence of parameter estimates $\hat{\theta}_k$ that satisfy the bound~\eqref{eq:hat_bound}, and consider the instantaneous and certainty-equivalent LQR update~\eqref{eq:lqr}. Then there exist constants $\alpha,\beta,\gamma > 0$ such that the following Lyapunov decrease condition is satisfied:
    \begin{equation} \label{eq:lyap_lqr}
        V(x_{k+1},\hat{\theta}_{k+1}) - V(x_k,\hat{\theta}_k) \leq -\alpha \| x_k\|^2 + \beta \| e_{1|k}\|^2 + \gamma \| w_k\|^2,\quad \forall k \in \mathbb{N},
    \end{equation}
    where $V(x,\theta)$ is defined in~\eqref{eq:Lyapunov_def}.
\end{proposition}
\begin{proof}
    The proof can be found in Appendix~\ref{app:prop}.
    \vspace{-0.2cm}
\end{proof}
\par The obvious limitation of directly applying Propositions~\ref{prop:lms} and~\ref{prop:lqr} is that we do not have an a~priori guarantee that the state and the input are bounded at all times. This is in contrast to prior works that utilize the LMS estimator for adaptive model predictive control~\citep{lorenzen2019robust, degner2026adaptive}, and overcoming this challenge is one of the key technical contributions of our work.
To this end, we define the Lyapunov function candidate for the joint state $(x_k,\varphi_k)$ as
\begin{equation} \label{eq:V_tilde_def}
        \tilde{V}(x_k,\varphi_k,k) := V(x_k,\varphi_k+\theta_k) + \frac{\beta}{\mu}\| \varphi_k\|^2 = x_k^\top P(\varphi_k+\theta_k) x_k + \frac{\beta}{\mu} \|\varphi_k\|^2,
\end{equation}
i.e., it is a linear combination of the Lyapunov functions from Propositions~\ref{prop:lms} and~\ref{prop:lqr}. Combining Propositions~\ref{prop:lms} and~\ref{prop:lqr} results in the main theoretical result of the paper.

\begin{theorem} \textbf{(Stability of the closed loop)} \label{thm}
Let Assumptions~\ref{ass:W} and~\ref{ass:Theta} hold, and assume that the diameter $d$ of the parameter set $\Theta$ is small enough. Consider the closed-loop interconnection of system~\eqref{eq:System}, projected LMS estimator~\eqref{eq:lms}, and certainty-equivalent LQR update~\eqref{eq:lqr}.
There exists a small enough step size $\mu > 0$ such that
\begin{enumerate}[a)]
    \item \label{enum:state_bound} the state and the input are uniformly bounded in time: $\|x_k\| \leq X,\ \|u_k\| \leq U,\ \forall k \in \mathbb{N}$ for some $X,U>0$,
    \item \label{enum:lyapunov} there exist constants $\tilde{\alpha},\tilde{\gamma},\tilde{\delta},\ubar{c},\bar{c}>0$ satisfying
    \begin{align}
        \underline{c} \| (x_k,\varphi_k)\|^2 \leq \tilde{V}(x_k,\varphi_k,k) &\leq \bar{c} \| (x_k,\varphi_k)\|^2,\quad \forall k \in \mathbb{N}, \label{eq:aug_lyap_bound}\\
        \tilde{V}(x_{k+1},\varphi_{k+1},k+1) - \tilde{V}(x_k,\varphi_k,k) &\leq - \tilde{\alpha}\| x_k\|^2 + \tilde{\gamma} \| w_k\|^2 + \tilde{\delta} \| \Delta\theta_k\|,\quad \forall k \in \mathbb{N}, \label{eq:aug_lyap_decrease} 
    \end{align}
    \item \label{enum:l2} and there exist constants $c_{1,2,3,4}>0$ such that $\forall T \in \mathbb{N}$ it holds that
    \begin{equation} \label{eq:x_l2}
        \sum_{k=0}^T \| x_k\|^2 \leq c_1 \|x_0\|^2+c_2\|\varphi_0\|^2 + c_3\sum_{k=0}^T \| w_k\|^2 + c_4\sum_{k=0}^T\| \Delta\theta_k\|,
    \end{equation}
    i.e., the state is finite-gain $\ell^2$-stable with respect to the disturbance $\{w\}$ as long as $\{\Delta\theta\} \in \ell^1$.
\end{enumerate}
\end{theorem}
\begin{proof}
    The proof can be found in Appendix~\ref{app:thm}.
    \vspace{-0.2cm}
\end{proof}
Inequality~\eqref{eq:x_l2} provides a finite-horizon upper bound on the state. Taking the limit $T \rightarrow \infty$, this result implies that the state trajectory is square-summable ($\{x\} \in \ell^2$), as long as $\{w\} \in \ell^2$ and $\{\Delta\theta\} \in \ell^1$, similarly to the corresponding result for the prediction error in Proposition~\ref{prop:lms}. 

\subsection{Discussion} \label{sec:discussion}

Compared to other non-episodic methods that learn the LQR online~\citep{wang2021exact, zhao2025policy}, the main advantage of the proposed method is its ability to handle continuously changing unknown parameters. Moreover, it does not need access to a stabilizing controller for the unknown system, instead, we require that the parameters lie inside the small enough known set $\Theta$. Additionally, other methods commonly use (recursive) least squares to estimate the parameter, and assume that the disturbance follows an independent and identically distributed zero-mean Gaussian distribution. In contrast, the LMS used by the proposed method is agnostic to distributional assumptions and instead requires boundedness of the disturbances (cf. Assumption~\ref{ass:W}).
\par Furthermore, existing methods (e.g., \citet{wang2021exact, zhao2025policy}) need to inject noise to ensure persistently exciting data.
In contrast, the proposed method lifts this requirement, allowing it to purely apply the LQR feedback without injecting additional noise.
Note that we require a small enough known parameter set $\Theta$.
Explicitly characterizing the exact condition on its size is difficult in practice. However, as also showcased in Section~\ref{sec:numerical_example}, this requirement is often satisfied by reasonable uncertainty sets in practical scenarios.
A resulting limitation of the proposed method is that it does not guarantee convergence of the parameter estimates. Consequently, the policy is also not guaranteed to converge to the optimal one.
On the other hand, this also highlights a further advantage:
stability of the closed-loop is guaranteed even without correctly learning the unknown system, which is beneficial whenever the unknown system is difficult-to-learn~\citep{tsiamis2021linear}. Finally, simulation results (see Section~\ref{sec:numerical_example}) suggest that, by ensuring persistent excitation, the proposed method is also able to learn the true parameter. This is not surprising, as the reason that the proposed method does not achieve perfect identification is that the LMS estimator only learns as long as there is a prediction error $e_{1|k}$. We expect that under persistent excitation $e_{1|k} \equiv 0$ implies the absence of parameter error, i.e., $\phi_k\equiv 0$.

\section{Numerical Example} \label{sec:numerical_example}


In order to demonstrate the real-world applicability of the proposed algorithm, we consider the stabilization of the planar quadrotor model used in~\citet{singh2023robust}, which is an unstable underactuated nonlinear system.\footnote{The code is available online: https://gitlab.ethz.ch/ics/adaptive-lqr-for-unknown-ltv-systems, \\\doi{10.3929/ethz-c-000798810}}
The discrete-time nonlinear dynamics obtained via forward Euler discretization can be written as

\small\begin{equation}  \label{eq:drone_nonlinear}
    x_{k+1} = \begin{bmatrix}
    p^x_{k+1} \\ p^z_{k+1} \\ \psi_{k+1} \\ v^x_{k+1} \\ v^z_{k+1} \\ \omega_{k+1} 
    \end{bmatrix} = \underbrace{\begin{bmatrix}
        1 & 0 & 0 & T_s & 0 & 0 \\ 0 & 1 & 0 & 0 & T_s & 0 \\ 0 & 0 & 1 & 0 & 0 & T_s \\
        0 & 0 & -gT_s & 1 & 0 & 0 \\
        0 & 0 & -[\theta_k]_1T_s & 0 & 1 & 0 \\ 0 & 0 & 0 & 0 & 0 & 1
    \end{bmatrix}}_{A(\theta_k)} x_k + \underbrace{\begin{bmatrix}
        0 & 0 \\ 0 & 0 \\ 0 & 0 \\ 0 & 0 \\ T_s/m & T_s/m \\ T_sl[\theta_k]_2 & -T_sl[\theta_k]_2 
    \end{bmatrix}}_{B(\theta_k)} u_k + 
    w_k^\mathrm{nl},
\end{equation}\normalsize
where the nonlinearities are captured by
\small\begin{equation}
    w_k^\mathrm{nl} = T_s\begin{bmatrix}
        v^x_k (\cos(\psi_k)-1)-v^z_k\sin(\psi_k) \\
        v^x_k \sin(\psi_k)+v^z_k(\cos(\psi_k)-1) \\
        0 \\
        v^z_k\omega_k - g\sin(\psi_k) + [\theta_k]_1 \cos(\psi_k) + g\psi_k \\
        -v^x_k\omega_k - g\cos(\psi_k) - [\theta_k]_1 \sin(\psi_k) + w^z_k + [\theta_k]_1\psi_k + [1/m\ \ 1/m]u_\mathrm{eq} \\
        w^\psi_k
    \end{bmatrix}.
\end{equation}\normalsize
The state $x \in \mathbb{R}^6$ consists of horizontal and vertical positions $p^x,p^z$, velocities in the body frame $v^x,v^z$, and pitch angle and angular velocity $\psi,\omega$, and the input $u \in \mathbb{R}^2$ denotes deviation from the hovering thrust $u_\mathrm{eq} = \frac{mg}{2}[1\ \ 1]^\top$. Disturbances $(w^z,w^\psi) \in [-1,1]^2$ represent inaccuracies in the actuation.
The unknown parameter vector $\theta = ([\theta]_1, [\theta]_2) \in \Theta = [-10,10] \times [50, 500]$ consists of a time-varying horizontal wind force ($[\theta]_1$), and the inverse of the moment of inertia of the out-of-plane axis ($[\theta]_2$). Note that the set $\Theta$ is of considerable size.
Parameters $g = 9.81$, $m = 0.5$, $l = 0.25$ corresponding to the gravitational acceleration, mass, and length  of one arm, respectively, are assumed to be known and constant, and the discretization time is $T_s = 0.1$ (all in SI units). It can be shown that $\forall\theta\in\Theta$ the pair $(A(\theta),B(\theta))$  is stabilizable, i.e., Assumption~\ref{ass:Theta} is satisfied.

\begin{figure}[h]
    \centering
    \includegraphics[width=0.9\textwidth]{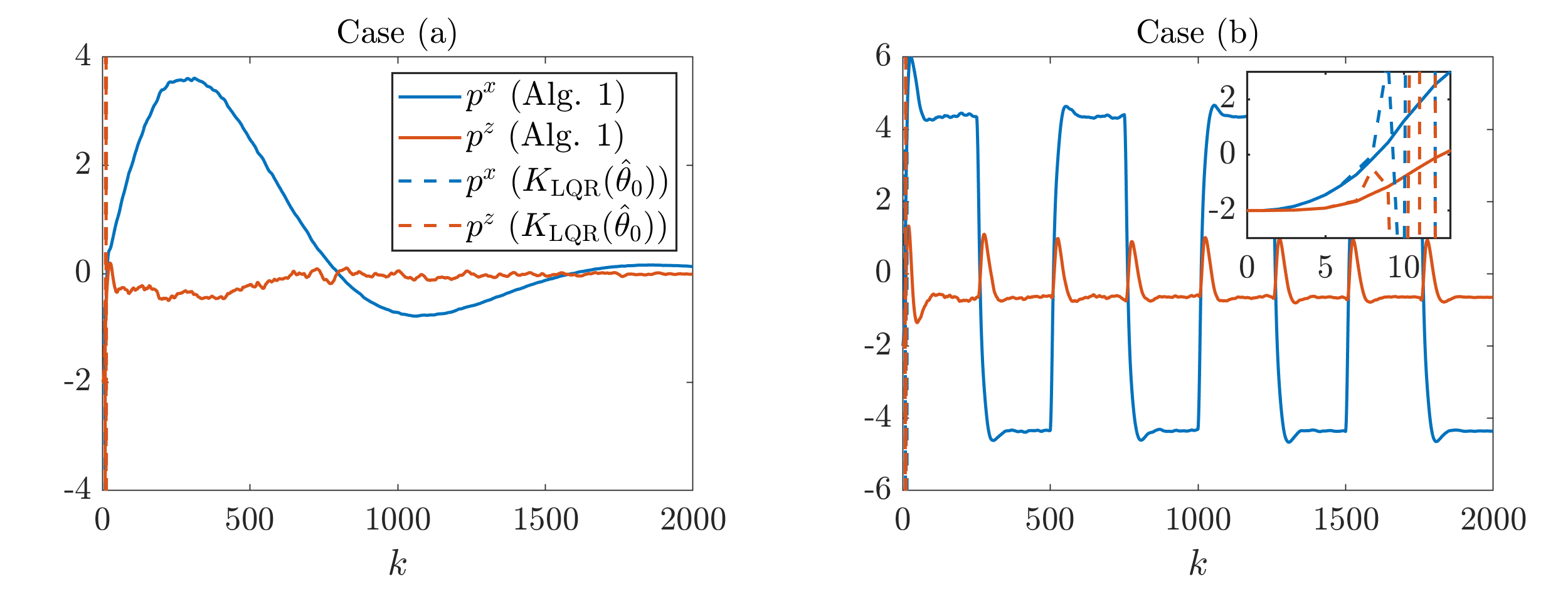}
    \caption{Closed-loop position trajectories of the proposed adaptive controller (Alg.~\ref{alg:LMS+LQR}) (solid) and a non-adaptive baseline (diverging dashed, see zoomed inset on the right) applied to the nonlinear planar quadrotor~\eqref{eq:drone_nonlinear} for Cases (a) and (b).} \label{fig:nonlinear}
    \vspace{-1em}
\end{figure}

\par Regarding the unknown parameters, we consider two cases. The wind $[\theta_k]_1$ follows a decaying
pattern in Case (a), and a persistent square waveform with constant amplitude in Case (b) (see Figure~\ref{fig:linear} for the wind profiles), and the true value for the inverse of the inertia is $[\theta_k]_2 \equiv 250$ in both cases. First, we apply Algorithm~\ref{alg:LMS+LQR} to control the true nonlinear system~\eqref{eq:drone_nonlinear}, with $\mu = 50$, $Q = I_6$, $R = 10I_2$.
The initial conditions are $p^x_0 = p^y_0 = -2$
and $\hat{\theta}_0 = [0\ \ 100]^\top$, and disturbances $(w^z_k,w^\psi_k)$ are sampled independently from a uniform distribution over $[-1,1]\cdot\mathrm{exp}(-0.001k)$.  As Figure~\ref{fig:nonlinear} shows, the proposed method achieves asymptotic stability of the closed-loop in Case (a), and keeps the system from diverging in Case (b).
Figure~\ref{fig:nonlinear} also highlights the need for online adaptation: the initial policy $K_\mathrm{LQR}(\hat{\theta}_0)$ is unable to stabilize the system.
\par In order to investigate the convergence of the parameter estimates, we test our algorithm on the linearized dynamics $x_{k+1} = A(\theta_k)x_k + B(\theta_k)u_k + w^\mathrm{lin}_k$ with $w_k^\mathrm{lin} = T_s[0\ \ 0\ \ 0 \ \ [\theta_k]_1 \ \ w^z_k \ \ w^\psi_k]^\top$. For this example, we apply the input $u_k = K_kx_k + \epsilon_k$, where the exploratory noise $\epsilon_k \stackrel{\text{i.i.d.}}{\sim} \mathcal{N}(0,\mathrm{diag}(0.5^2,0.1^2))$ is injected to the system to ensure persistency of excitation. As illustrated by Figure~\ref{fig:linear}, the proposed method is able to learn the true value of the moment of inertia and track the unknown time-varying wind parameter. This highlights a trade-off that one can make: if only stability is important, then there is no need to inject noise, but learning the parameter can also be achieved by injecting additional noise.

\begin{figure}[h]
    \centering
    \includegraphics[width=0.9\textwidth]{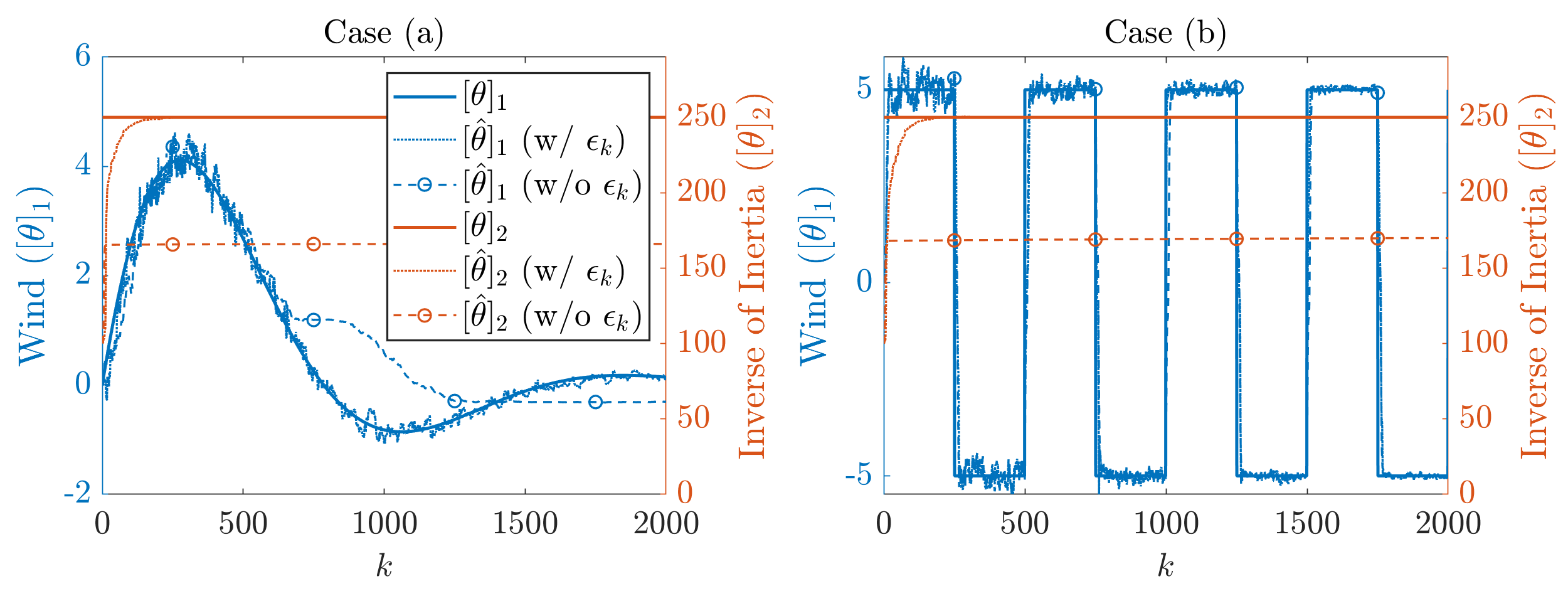}
    \caption{Comparison of the proposed adaptive controller (Alg.~\ref{alg:LMS+LQR}) with (dotted) and without (dashed) additional exploratory noise $\epsilon_k$ in terms of convergence of the estimates to the true values (solid) on the linearized dynamics for Cases (a) and (b).} \label{fig:linear}
    \vspace{-2em}
\end{figure}

\section{Conclusion and Outlook} \label{sec:conclusion}

In this work, we proposed an adaptive control scheme for discrete-time linear systems with unknown time-varying parameters. The proposed method is a combination of two off-the-shelf building blocks from the literature: the certainty-equivalent linear quadratic regulator and the projected least mean square estimator. Despite its simplicity, $\ell^2$-stability of the closed-loop interconnection was proven, even without assuming that the data is persistently exciting. The real-world applicability of the proposed method was demonstrated on a nonlinear planar quadrotor model. Furthermore, simulation results also suggest that the proposed method is able to learn the unknown parameter if the data is persistently exciting, resulting in asymptotically optimal behavior. The formal investigation of this hypothesis is the subject of future work.

\appendix
\section{Technical Proofs}

\subsection{Notation and Preliminaries} \label{app:notation}
We will often use the shorthand notations $A_k := A(\theta_k), B_k := B(\theta_k), \hat{A}_k := A(\hat{\theta}_k), \hat{B}_k := B(\hat{\theta}_k)$, $P_k := P(\hat{\theta}_k), K_k := K_\mathrm{LQR}(\hat{\theta}_k)$, $\hat{A}_{\mathrm{cl},k} := \hat{A}_k+\hat{B}_kK_k$, and $\tilde{Q} :=\mathrm{diag}(Q,R)$. Since $\Theta$ is bounded and pointwise stabilizable (cf. Assumption~\ref{ass:Theta}), we can define the following uniform bounds:
\begin{align}
\begin{split} \label{eq:uniform_const}
    &\| K_\mathrm{LQR}(\theta)\| \leq \bar{K}, \quad P(\theta) \preceq \bar{p}I, \quad R \preceq rI, \quad  \ubar{q} \preceq Q \preceq qI,\quad \forall\theta \in \Theta, \\ &\| A(\theta)+B(\theta)K_\mathrm{LQR}(\theta')\| \leq \bar{a},\quad \forall\theta,\theta' \in \Theta.
\end{split}
\end{align}
\par Throughout the proofs we will make use of the following lemma (in addition to Lemma~\ref{lem:lqr_lipschitz}):
\begin{lemma}\textbf{(Young's inequality)} \label{lem:young}
    For any $0 \prec M\in \mathbb{R}^{n\times n},\ \forall\varepsilon>0$ and $\forall x,y \in \mathbb{R}^n$, it holds that
    \begin{equation*}
        \|x+y\|^2_M \leq (1 + \varepsilon)\| x\|^2_M + \left(1 + \frac{1}{\varepsilon}\right) \|y\|^2_M.
    \end{equation*}
\end{lemma}
We will also use the following convergent series:
\begin{equation} \label{eq:polylog}
    \sum_{k=0}^\infty q^k = \frac{1}{1-q}, \quad \sum_{k=0}^\infty kq^{k-1} = \frac{1}{(1-q)^2}, \quad \sum_{k=0}^\infty k^2q^{k-1} = \frac{1+q}{(1-q)^3},\quad \forall |q| < 1.
\end{equation}
Finally, we remark that, since $A(\theta)$ and $B(\theta)$ are affine, Lemma~\ref{lem:lqr_lipschitz} implies that $\exists L_A > 0$ such that
\begin{equation} \label{eq:Acl_lipschitz}
    \|\hat{A}_{\mathrm{cl},j} - \hat{A}_{\mathrm{cl},k} \| \leq L_A\| \hat{\theta}_j-\hat{\theta}_k\|,\quad \forall \hat{\theta}_j, \hat{\theta}_k \in \Theta.
\end{equation}

\subsection{Proof of Lemma~\ref{lem:lqr_lipschitz}} \label{app:lqr_lipschitz}
\begin{proof}
    As shown in~\citet[Proposition~2]{east2020infinite}, the Jacobians $\partial \mathrm{vec}(P(\theta))/\partial\mathrm{vec}(A(\theta))$ and $\partial \mathrm{vec}(P(\theta))/\partial\mathrm{vec}(B(\theta))$ exist and are given by
    \begin{equation*}
        \frac{\partial \mathrm{vec}(P(\theta))}{\partial\mathrm{vec}(A(\theta))} = Z_1^{-1}Z_2, \quad \frac{\partial \mathrm{vec}(P(\theta))}{\partial\mathrm{vec}(B(\theta))} = Z_1^{-1}Z_3,
    \end{equation*}
    as long as $Z_1$ is invertible, where
    \begin{align*}
        Z_1 &= I_{n^2} - (A^\top\otimes A^\top)\big( I_{n^2} - (PBM_2B^\top \otimes I_n) - (I_n \otimes PBM_2B^\top ) \\
        &\hspace{15em} +(PB \otimes PB)(M_2 \otimes M_2)(B^\top \otimes B^\top)\big), \\
        Z_2 &= (V_{n,n} + I_{n^2})(I_{n^2} \otimes A^\top M_1), \\
        Z_3 &= (A^\top \otimes A^\top)\big( (PB \otimes PB)(M_2 \otimes M_2)(I_{m^2}-V_{m,m}(I_m \otimes B^\top P))\big) \\
        &\hspace{20em} -(I_{n^2}+V_{n,n}(PBM_2 \otimes P)), \\
        &\quad M_1 = P - PBM_2B^\top P, \quad M_2 = M_3^{-1}, \quad M_3 = R + B^\top P B,
    \end{align*}
    where the permutation matrix $V_{n,m} \in \mathbb{R}^{nm\times nm}$ is implicitly defined by $V_{n,m} \mathrm{vec}(A) = \mathrm{vec}(A^\top)$, and we omit the argument $\theta$ for ease of notation. Since
    \begin{equation*}
        Z_1 = I_{n^2} - (A^\top - A^\top PBM_2B^\top) \otimes (A^\top - A^\top PBM_2B^\top) = I_{n^2} - (A+BK)^\top \otimes (A+BK)^\top,
    \end{equation*}
    and $A+BK$ is Schur stable, $Z_1$ is invertible over $\Theta$, moreover, there exist $C_Z>0$ and $0\leq\rho_Z<1$ such that
    \begin{equation*}
        \|((A+BK)^\top \otimes (A+BK)^\top)^k \| \leq C_Z\rho_Z^k,\ \forall k \in \mathbb{N},\ \forall \theta \in \Theta.
    \end{equation*}
    Therefore,
    \begin{align*}
        \|Z_1^{-1}\| &= \|(I_{n^2} - (A+BK)^\top \otimes (A+BK)^\top)^{-1} \| \\ &\stackrel{\text{Neumann series}}{\leq} \sum_{k=0}^\infty \|((A+BK)^\top \otimes (A+BK)^\top)^k \| \leq \sum_{k=0}^\infty C_Z\rho_Z^k = \frac{C_Z}{1-\rho_Z},
    \end{align*}
    which, combined with compactness of $\Theta$, proves that the Jacobians are uniformly upper bounded over $\Theta$.
    \par Based on this, the Jacobians $\partial \mathrm{vec}(K_\mathrm{LQR}(\theta))/\partial\mathrm{vec}(A(\theta))$ and $\partial \mathrm{vec}(K_\mathrm{LQR}(\theta))/\partial\mathrm{vec}(B(\theta))$ can be derived using the same techniques as used in~\citet{east2020infinite}, and the corresponding uniform upper bounds also exist. Since $A(\theta)$ and $B(\theta)$ are affine in $\theta$, this implies that the Jacobian $\partial \mathrm{vec}(K_\mathrm{LQR}(\theta)) / \partial \theta$ is uniformly upper bounded over  $\Theta$, proving Lipschitz continuity over $\Theta$.
\end{proof}

\subsection{Proof of Proposition~\ref{prop:lqr}} \label{app:prop}
\begin{proof}
In order to bound the difference $V(x_{k+1},\hat{\theta}_{k+1}) - V(x_k,\hat{\theta}_k)$, we split it into two parts:
\begin{equation*}
    V(x_{k+1},\hat{\theta}_{k+1}) - V(x_k,\hat{\theta}_k) = \underbrace{V(x_{k+1},\hat{\theta}_{k+1}) - V(x_{k+1},\hat{\theta}_k)}_{(*)} + \underbrace{V(x_{k+1},\hat{\theta}_k) - V(x_k,\hat{\theta}_k)}_{(**)}.
\end{equation*}
\par The difference $(**)$ can be upper bounded as follows:
\begin{align*}
    V(x_{k+1},\hat{\theta}_k) - V(x_k,\hat{\theta}_k) &= \|x_{k+1}\|^2_{P_k} - \|x_k\|^2_{P_k} \\
    &\leq \left(1+\frac{1}{\varepsilon_1}\right)\| e_{1|k} + w_k\|^2_{P_k} + (1+\varepsilon_1)\| \hat{x}_{1|k}\|^2_{P_k} - \| x_k\|^2_{P_k},
\end{align*}
for arbitrary $\varepsilon_1>0$, due to $x_{k+1} = \hat{x}_{1|k} + e_{1|k} + w_k$ and Lemma~\ref{lem:young}. As $P_k$ solves~\eqref{eq:dare},
\begin{equation*}
    \|\hat{x}_{1|k}\|^2_{P_k} =  \|(\hat{A}_k+\hat{B}_kK_k)x_k\|^2_{P_k} = \|x_k\|^2_{P_k} - \|x_k\|^2_{Q+K_k^\top RK_k},
\end{equation*}
which results in the following upper bound for $(**)$:
\begin{align}
    V(x_{k+1},\hat{\theta}_k) - V(x_k,\hat{\theta}_k) &\leq \left(1+\frac{1}{\varepsilon_1}\right)\| e_{1|k} + w_k\|^2_{P_k} + \varepsilon_1\| x_k\|^2_{P_k}- (1+\varepsilon_1)\| x_k \|^2_{Q+K_k^\top RK_k} \nonumber \\
    &\stackrel{\eqref{eq:uniform_const}}{\leq}  - ((1+\varepsilon_1)\ubar{q} - \varepsilon_1\bar{p})\| x_k \|^2 + \left(1+\frac{1}{\varepsilon_1}\right)\bar{p}\| e_{1|k} + w_k\|^2\label{eq:**final}.
\end{align}
\par In order to bound $(*)$, note that $V(x,\theta)$ is the infinite-horizon cost-to-go corresponding to system $A(\theta),B(\theta)$ under controller $K_\mathrm{LQR}(\theta)$:
\begin{equation} \label{eq:infinite_cost_to_go}
    V(x,\theta) = x^\top P(\theta) x = \sum_{j=0}^\infty \| (x_j,K_\mathrm{LQR}(\theta)x_j)\|^2_{\tilde{Q}},
\end{equation}
where $x_{j+1} = (A(\theta)+B(\theta)K_\mathrm{LQR}(\theta))x_j,\ x_0 = x$. To express $V(x_{k+1},\hat{\theta}_{k+1})$ and $V(x_{k+1},\hat{\theta}_k)$ via~\eqref{eq:infinite_cost_to_go}, define the auxiliary systems $z_{i+1} = \hat{A}_{\mathrm{cl},k+1}z_i$ and $y_{i+1}=\hat{A}_{\mathrm{cl},k}y_i$ initialized at $z_0=y_0=x_{k+1}$, and let 
\begin{equation} \label{eq:xi_eta_def}
    \xi_i :=z_i-y_i, \qquad\eta_i :=K_{k+1}z_i - K_ky_i.
\end{equation}
Then,
\begin{align*}
    \eta_i &= K_{k+1}\xi_i + (K_{k+1}-K_k)y_i\\
    \xi_{i+1} &= \hat{A}_{\mathrm{cl},k+1}\xi_i + (\hat{A}_{\mathrm{cl},k+1}-\hat{A}_{\mathrm{cl},k})y_i,
\end{align*}
that is,
\begin{align*}
    \xi_j &= \sum_{i=0}^{j-1}\hat{A}_{\mathrm{cl},k+1}^i(\hat{A}_{\mathrm{cl},k+1}-\hat{A}_{\mathrm{cl},k})\hat{A}_{\mathrm{cl},k}^{j-1-i}x_{k+1}, \\
    \eta_j &= (K_{k+1}-K_k)\hat{A}_{\mathrm{cl},k}^jx_{k+1} + K_{k+1}\sum_{i=0}^{j-1}\hat{A}_{\mathrm{cl},k+1}^i(\hat{A}_{\mathrm{cl},k+1}-\hat{A}_{\mathrm{cl},k})\hat{A}_{\mathrm{cl},k}^{j-1-i}x_{k+1},
\end{align*}
implying
\begin{align}
    \| \xi_j \| &\leq \|\hat{A}_{\mathrm{cl},k+1}-\hat{A}_{\mathrm{cl},k}\|\sum_{i=0}^{j-1}\|\hat{A}_{\mathrm{cl},k+1}^i\|\|\hat{A}_{\mathrm{cl},k}^{j-1-i}\| \|x_{k+1}\|, \label{eq:xi_bound_1}\\
    \| \eta_j \| &\leq \|K_{k+1}-K_k\|\|\hat{A}_{\mathrm{cl},k}^j\|\|x_{k+1}\|  \label{eq:eta_bound_1}\\
    &\quad +\|K_{k+1}\|\|\hat{A}_{\mathrm{cl},k+1}-\hat{A}_{\mathrm{cl},k}\|\sum_{i=0}^{j-1}\|\hat{A}_{\mathrm{cl},k+1}^i\|\|\hat{A}_{\mathrm{cl},k}^{j-1-i}\|\|x_{k+1}\|, \nonumber
\end{align}
Since $\hat{A}_{\mathrm{cl},k}$ and $\hat{A}_{\mathrm{cl},k+1}$ are Schur stable, there exist $C>0$ and $0\leq\rho<1$ such that
\begin{equation*}
    \| \hat{A}_{\mathrm{cl},k}^i\| \leq C\rho^i,\quad \| \hat{A}_{\mathrm{cl},k+1}^i\| \leq C\rho^i,\quad \forall i\geq0, \quad \forall\hat{\theta}_k \in \Theta,
\end{equation*}
resulting in
\begin{equation*}
    \sum_{i=0}^{j-1}\|\hat{A}_{\mathrm{cl},k+1}^i\|\|\hat{A}_{\mathrm{cl},k}^{j-1-i}\| \leq \sum_{i=0}^{j-1}C^2\rho^{j-1} = C^2j\rho^{j-1}.
\end{equation*}
Consequently,
\begin{align}
    \| \xi_j \| &\stackrel{\eqref{eq:xi_bound_1}}{\leq} C^2\|\hat{A}_{\mathrm{cl},k+1}-\hat{A}_{\mathrm{cl},k}\|\|x_{k+1}\|j\rho^{j-1}, \label{eq:xi_bound_2}\\
    \| \eta_j \| &\stackrel{\eqref{eq:eta_bound_1}}{\leq} C^2\|K_{k+1}-K_k\|\|x_{k+1}\|\rho^j +C^2\|K_{k+1}\|\|\hat{A}_{\mathrm{cl},k+1}-\hat{A}_{\mathrm{cl},k}\|\|x_{k+1}\|j\rho^{j-1}. \label{eq:eta_bound_2}
\end{align}
Note that
\begin{equation}\label{eq:xk+1_bound}
    \|x_{k+1}\|^2  = \| (A_k+B_kK_k)x_k + w_k\|^2 \stackrel{\text{Lemma \ref{lem:young}}}{\leq} 2\|(A_k+B_kK_k)x_k \|^2 + 2\|w_k\|^2 \leq2\bar{a}^2X^2+2\bar{W}^2,
\end{equation}
where the last inequality is due to~\eqref{eq:uniform_const} and the assumption that $\|x_k\| \leq X$. Taking the squares of~\eqref{eq:xi_bound_2}-\eqref{eq:eta_bound_2}, and using $\rho^j\leq \rho^{j-1}$, Lemma~\ref{lem:lqr_lipschitz}, and~\eqref{eq:uniform_const},~\eqref{eq:Acl_lipschitz}, and~\eqref{eq:xk+1_bound} results in
\begin{align}
\label{eq:xi_eta_bound_3}
    \| \xi_j \|^2 &\leq 2C^2(\bar{a}^2X^2 +\bar{W}^2)L_A^2j^2(\rho^2)^{j-1}\|\hat{\theta}_{k+1}-\hat{\theta}_k\|^2 ,\\
    \| \eta_j \|^2 &\leq 2C^2(\bar{a}^2X^2 +\bar{W}^2)\left(L_K^2(\rho^2)^j + 2\bar{K}L_KL_Aj(\rho^2)^{j-1}+ \bar{K}^2L_A^2j^2(\rho^2)^{j-1}\right)\|\hat{\theta}_{k+1}-\hat{\theta}_k\|^2. \nonumber
\end{align}
From~\eqref{eq:infinite_cost_to_go}, it can be seen that $\forall \varepsilon_2>0$
\begin{align}
    V(x_{k+1},\hat{\theta}_{k+1}) - V(x_{k+1},\hat{\theta}_k) &= \sum_{j=0}^\infty \left(\| (z_j,K_{k+1}z_j)\|^2_{\tilde{Q}}-\| (y_j,K_ky_j))\|^2_{\tilde{Q}}\right) \nonumber \\
     &\stackrel{\eqref{eq:xi_eta_def}}{=} \sum_{j=0}^\infty \left( \| (y_j+\xi_j,K_ky_j+\eta_j)\|^2_{\tilde{Q}}-\| (y_j,K_ky_j))\|^2_{\tilde{Q}} \right) \nonumber \\
     & \stackrel{\text{Lemma \ref{lem:young}}}{\leq}\sum_{j=0}^\infty \left(\varepsilon_2\| (y_j,K_ky_j))\|^2_{\tilde{Q}} + \left(1+\frac{1}{\varepsilon_2}\right)\| (\xi_j,\eta_j)\|^2_{\tilde{Q}} \right) \nonumber \\
    &= \varepsilon_2\| x_{k+1}\|^2_{P_k} + \left(1+\frac{1}{\varepsilon_2}\right)\sum_{j=0}^\infty(\|\xi_j\|^2_Q + \| \eta_j\|_R^2 ) \nonumber \\
    &\stackrel{\eqref{eq:uniform_const}}{\leq} \varepsilon_2\bar{p}\| x_{k+1}\|^2 + \left(1+\frac{1}{\varepsilon_2}\right)\sum_{j=0}^\infty(q\|\xi_j\|^2 + r\| \eta_j\|^2 ). \label{eq:V_diff_intermediate}
\end{align}
Combining~\eqref{eq:xi_eta_bound_3} with the infinite sums~\eqref{eq:polylog} yields:
\begin{align}
\begin{split} \label{eq:xi_eta_bound_4}
    \sum_{j=0}^\infty\|\xi_j\|^2 &\leq 2C^2(\bar{a}^2X^2 + \bar{W}^2)\|\hat{\theta}_{k+1}-\hat{\theta}_k\|^2 L_A^2\frac{1+\rho^2}{(1-\rho^2)^3}, \\
    \sum_{j=0}^\infty\|\eta_j\|^2 &\leq 2C^2(\bar{a}^2X^2 + \bar{W}^2)\|\hat{\theta}_{k+1}-\hat{\theta}_k\|^2 \left(\frac{L_K^2}{1-\rho^2} + \frac{2\bar{K}L_KL_A}{(1-\rho^2)^2} +\bar{K}^2L_A^2 \frac{1+\rho^2}{(1-\rho^2)^3}\right).
\end{split}
\end{align}
Applying Lemma~\ref{lem:young} with $\varepsilon_3>0$ on $\|x_{k+1}\|^2 = \|\hat{x}_{1|k} + e_{1|k} + w_k\|^2$ gives
\begin{equation} \label{eq:varepsilon3}
    \|x_{k+1}\|^2 \leq \left(1 + \frac{1}{\varepsilon_3}\right)\|e_{1|k} + w_k\|^2 + (1 + \varepsilon_3)\bar{a}^2\|x_k\|^2.
\end{equation}
Finally, combining~\eqref{eq:V_diff_intermediate},~\eqref{eq:xi_eta_bound_4},~\eqref{eq:varepsilon3}, and~\eqref{eq:hat_bound} results in
\begin{align} \label{eq:*final}
    V(x_{k+1}&,\hat{\theta}_{k+1}) - V(x_{k+1},\hat{\theta}_k)  \\ \leq &\left(\left(1+\frac{1}{\varepsilon_2}\right)\mu (C_1X^2+C_2\bar{W}^2) + \varepsilon_2\left(1+\frac{1}{\varepsilon_3}\right)\bar{p}\right)\| e_{1|k} + w_k\|^2 + \varepsilon_2(1+\varepsilon_3) \bar{a}^2\bar{p} \| x_k\|^2, \nonumber
\end{align}
for some $C_1,C_2>0$.

\par Summing~\eqref{eq:*final} and~\eqref{eq:**final} leads to
\begin{align} 
\label{eq:prop2_final}
    V(x_{k+1},\hat{\theta}_{k+1}) &- V(x_k,\hat{\theta}_k) \leq -((1+\varepsilon_1)\ubar{q}-\varepsilon_2(1+\varepsilon_3) \bar{a}^2\bar{p} -\varepsilon_1\bar{p}) \| x_k\|^2 \\ &+ \left(\left(1+\frac{1}{\varepsilon_2}\right) \mu (C_1X^2 +C_2\bar{W}^2)+ \left(\varepsilon_2\left(1+\frac{1}{\varepsilon_3}\right) + 1+\frac{1}{\varepsilon_1}\right)\bar{p}\right)\| e_{1|k} + w_k\|^2, \nonumber
\end{align}
$\forall \varepsilon_{1,2,3}>0$. By choosing $\varepsilon_{1,2,3}>0$ small enough and once again using Lemma~\ref{lem:young} on $\| e_{1|k}+w_k\|^2$, we obtain $\alpha,\beta,\gamma>0$ satisfying the Lyapunov decrease condition~\eqref{eq:lyap_lqr},
where $\beta$ is of the form
\begin{equation} \label{eq:beta_form}
    \beta = \nu_2\mu X^2 + \nu_3\mu + \nu_4
\end{equation}
with $\nu_{2,3,4} > 0$ constants independent of $X$ and $\mu$.
\end{proof}

\subsection{Proof of Theorem~\ref{thm}} \label{app:thm}
\begin{proof}

    First, we address item~\ref{enum:lyapunov} of Theorem~\ref{thm} by showing that the candidate Lyapunov function~\eqref{eq:V_tilde_def} satisfies~\eqref{eq:aug_lyap_bound}-\eqref{eq:aug_lyap_decrease} for a given $k$, assuming $\|x_k\| \leq X$ for some $X \geq 1$. Then, we show that sublevel sets of the Lyapunov function are invariant for the joint state $(x_k,\varphi_k)$. Lastly, we prove that if $d>0$ (the diameter of $\Theta$) is small enough, there is a sufficiently large $X>0$ such that containment in the invariant set implies $\|x_{k+1}\| \leq X$, proving items~\ref{enum:state_bound}-\ref{enum:lyapunov} of Theorem~\ref{thm}. Finally, we prove item~\ref{enum:l2}.

    \par If $\|x_k\| \leq X$, then $\| u_k\| = \| K(\hat{\theta}_k)x_k\|$ is also bounded, since $\hat{\theta}_k \in \Theta$. Thus, by selecting $\mu = \nu_1/X^2$, where $\nu_1>0$ is some constant depending on the known functions $A(\theta)$ and $B(\theta)$, the step size $\mu>0$ satisfies~\eqref{eq:mu_cond}, ensuring that the conditions in Propositions~\ref{prop:lms} and~\ref{prop:lqr} hold for $k$, resulting in inequalities~\eqref{eq:lyap_phi} and \eqref{eq:lyap_lqr} for $k$.
    Multiplying~\eqref{eq:lyap_phi} by $\beta/\mu > 0$ and adding it to~\eqref{eq:lyap_lqr} yields
    \begin{align}
    \label{eq:V_tilde_decrease}
        \tilde{V}(x_{k+1},\varphi_{k+1},k+1) - \tilde{V}(x_k,\varphi_k,k) \leq-\alpha \| x_k\|^2 + (\beta + \gamma) \| w_k\|^2 + \frac{2d\beta}{\mu}\| \Delta\theta_k \|,
    \end{align}proving~\eqref{eq:aug_lyap_decrease} with $\tilde{\alpha} = \alpha$, $\tilde{\gamma} = \gamma+\beta$, $\tilde{\delta} = 2d\beta/\mu$. It can also be seen that $\tilde{V}(x_k,\varphi_k,k)$ satisfies the uniform lower and upper bounds~\eqref{eq:aug_lyap_bound}:
    \begin{equation} \label{eq:uniform_bounds}
        \ubar{q}\| x_k\|^2 + \frac{\beta}{\mu}\| \varphi_k\|^2 \leq \tilde{V}(x_k,\varphi_k,k) \leq \bar{p}\| x_k\|^2 + \frac{\beta}{\mu}\| \varphi_k\|^2.
    \end{equation}
    \par Define the time-varying sublevel set of $\tilde{V}(x_k,\varphi_k,k)$ as
    \begin{equation*}
         \mathcal{S}_k := \{\ (x_k,\varphi_k)\ |\ \tilde{V}(x_k,\varphi_k,k) \leq b\ \},
    \end{equation*}
    where
    \begin{equation*}
        b := \max \left\{\bar{p}\| x_0\|^2; \left(\frac{\bar{p}}{\alpha}+1\right)\left((\gamma + \beta)\bar{W}^2 + 2\frac{\beta}{\mu}d^2\right)\right\} + \frac{\beta}{\mu}d^2,
    \end{equation*}
    ensuring that $(x_0,\varphi_0) \in \mathcal{S}_0$.
    We use case distinction to show that $(x_k,\varphi_k) \in \mathcal{S}_k \implies (x_{k+1},\varphi_{k+1}) \in \mathcal{S}_{k+1}$:
    \begin{enumerate}[leftmargin=2.5\parindent]
        \item[Case 1:] if $\|x_k\|^2 \geq (\gamma + \beta)\frac{\bar{W}^2}{\alpha} + \frac{2\beta d^2}{\mu\alpha}$, then, noting that $\|\Delta\theta_k\| = \| \theta_{k+1} - \theta_k\| \leq d$, it holds that
        \begin{equation*}
            \tilde{V}(x_{k+1},\varphi_{k+1},k+1) - \tilde{V}(x_k,\varphi_k,k) \stackrel{\eqref{eq:V_tilde_decrease}}{\leq} -\alpha \| x_k\|^2 + (\beta + \gamma) \| w_k\|^2 + \frac{2d\beta}{\mu}\| \Delta\theta_k \| \leq 0,
        \end{equation*}
        implying $\tilde{V}(x_{k+1},\varphi_{k+1},k+1) \leq \tilde{V}(x_k,\varphi_k,k) \leq b$, i.e., $(x_{k+1},\varphi_{k+1}) \in \mathcal{S}_{k+1}$.
        \item[Case 2:] if $\|x_k\|^2 < (\gamma + \beta)\frac{\bar{W}^2}{\alpha} + \frac{2\beta d^2}{\mu\alpha}$, then
        \begin{align*}
            \tilde{V}(x_{k+1},\varphi_{k+1},k+1) &\stackrel{\eqref{eq:V_tilde_decrease}}{\leq} \tilde{V}(x_k,\varphi_k,k) -\alpha \| x_k\|^2 + (\beta + \gamma) \| w_k\|^2 + \frac{2d\beta}{\mu}\| \Delta\theta_k \| \\
            &\stackrel{\eqref{eq:uniform_bounds}}{\leq} \bar{p}\| x_k\|^2 + \frac{\beta}{\mu}d^2 + (\beta + \gamma) \bar{W}^2 + \frac{2\beta }{\mu}d^2 \\
            &< \bar{p}((\gamma + \beta)\frac{\bar{W}^2}{\alpha} + \frac{2\beta }{\mu\alpha}d^2) + \frac{\beta}{\mu}d^2 + (\beta + \gamma) \bar{W}^2 + \frac{2\beta }{\mu}d^2 \\
            &= \left(\frac{\bar{p}}{\alpha}+1\right)\left((\gamma + \beta)\bar{W}^2 + 2\frac{\beta}{\mu}d^2\right) + \frac{\beta}{\mu}d^2 \leq b,
        \end{align*}
        i.e., $(x_{k+1},\varphi_{k+1}) \in \mathcal{S}_{k+1}$.
    \end{enumerate}
    That is $(x_k,\varphi_k) \in \mathcal{S}_k \implies (x_{k+1},\varphi_{k+1}) \in \mathcal{S}_{k+1}$, further implying $\| x_{k+1}\|^2 \leq b/\ubar{q}$ due to~\eqref{eq:uniform_bounds}.
    Note that in order for Proposition~\ref{prop:lqr} to hold, $\beta$ needs to satisfy $\beta \geq\nu_2\mu X^2 + \nu_3\mu + \nu_4$ with $\nu_{2,3,4} > 0$ (see~\eqref{eq:beta_form}). Since $\mu = \nu_1/X^2$ and $X \geq 1$, this means that $\beta = \nu_1(\nu_2 + \nu_3) + \nu_4 =: \nu_5$ is a valid choice, in which case $\beta/\mu = \nu_5X^2/\nu_1$. As $\alpha$ is independent of $X$, $\mu$, and $d$, and we are free to select $\gamma = \beta$, $b$ is of the form 
    \begin{equation*}
        b = \max \{ \bar{p} \|x_0\|^2; \nu_6 + \nu_7d^2X^2\} + \nu_8 d^2 X^2,
    \end{equation*}
    with constants $\nu_{6,7,8}>0$ independent of $d$ and $X$. 
    Consequently, to ensure that $(x_k,\varphi_k) \in \mathcal{S}_k$ implies $\|x_k\| \leq X$, it must hold that
        \begin{equation*}
        \ubar{q}X^2 \geq b =\max \{ \bar{p} \|x_0\|^2; \nu_6 + \nu_7d^2X^2\} + \nu_8 d^2 X^2,
    \end{equation*}
    which has a solution $X\geq 1$ as long as $d$ is small enough.
    \par In summary, for sufficiently small $d$, there exists a large enough choice for $X$ such that $(x_k,\varphi_k) \in \mathcal{S}_k$ implies both $\|x_k\| \leq X$ and $(x_{k+1},\varphi_{k+1}) \in \mathcal{S}_{k+1}$. As $(x_0,\varphi_0) \in \mathcal{S}_0$ due to~\eqref{eq:uniform_bounds}, we have proven item~\ref{enum:lyapunov} of Theorem~\ref{thm} by induction. Additionally, $\|x_k\| \leq X,\ \forall k \in \mathbb{N}$ implies $\|u_k\| \leq U,\ \forall k \in \mathbb{N}$ for some $U>0$, proving item~\ref{enum:state_bound}. Finally, regarding item~\ref{enum:l2}, bound~\eqref{eq:x_l2} can be proven by summing up~\eqref{eq:aug_lyap_decrease} from $k=0$ to $k=T-1$, simplifying the resulting telescopic sum, and using bounds~\eqref{eq:aug_lyap_bound} on $\tilde{V}(x_T,\varphi_T,T)$ and $\tilde{V}(x_0,\varphi_0,0)$.
\end{proof}

\acks{This work was supported as a part of NCCR Automation, a National Centre of Competence in Research, funded by the Swiss National Science Foundation (grant number 51NF40\textunderscore225155).}

\bibliography{bibliography}

\end{document}